\documentclass[10pt, conference, letterpaper]{ IEEEtran}
\usepackage{cite}
\usepackage{booktabs}
\usepackage{graphicx}
\usepackage{balance}
\usepackage{epsfig}
\usepackage{subcaption} 
\usepackage{epstopdf}
\usepackage{enumitem}
\usepackage{wrapfig}
\usepackage{amsmath}
\usepackage{amssymb}
\usepackage{amsthm}
\usepackage{algorithm}
\usepackage{algpseudocode}
\newtheorem{theorem}{Theorem}
\usepackage{bm}
\usepackage{url}
\usepackage{amsfonts}
\newcommand{\m}[1]{\mathcal{#1}}

\ifCLASSINFOpdf
\else 
\fi
\usepackage{colortbl}

\usepackage{lipsum}
\ifCLASSOPTIONcompsoc
    \usepackage[caption=false, font=normalsize, labelfont=sf, textfont=sf]{subfig}
\else
\usepackage[caption=false, font=footnotesize]{subfig}
\fi
\usepackage{soul}
\usepackage{multirow}
\usepackage{booktabs}
\usepackage{tabularx}
\usepackage{tikz}
\usepackage{tikzpagenodes}
\usetikzlibrary{calc}

\AddToHook{shipout/foreground}{%
	\ifnum\value{page}=1
	\begin{tikzpicture}[remember picture, overlay]
		\node[
		draw,
		text width=0.95\textwidth,
		align=left,
		font=\footnotesize,
		inner sep=5pt
		]
		at ($(current page header area.center) + (0,-5pt)$)
		{%
			This paper has been accepted for publication in the proceedings
			of the IEEE International Conference on Intelligent Systems and Interdisciplinary Applications (IEEE ISAIA 2026). This is the authors' accepted version
			of the article. The final version published by IEEE is: A.~Alchaab, A.~Younis and D.~Pompili, ``FedRIC: Heterogeneity-Aware Federated Learning with Hard Deadline Constraints in Open RAN,'' in \textit{Proc. of
				IEEE International Conference on Intelligent Systems and Interdisciplinary Applications (IEEE ISAIA)}, Montclair, US, October 9, 2026.
		};
	\end{tikzpicture}%
	\fi
}

\begin{document}
\title{{FedRIC: Heterogeneity-Aware Federated Learning with Hard Deadline Constraints in Open RAN
%FedRIC: Heterogeneity-Aware Federated Learning with Hard Deadline Services in O-RAN RIC
}\vspace{-.1in}}
\author{\textbf{Adhwaa Al-Chaab†, Ayman Younis, Dario Pompili†}\\{†Department of Electrical and Computer Engineering}\\
{Rutgers University--New Brunswick, NJ, USA}\\
{E-mails: \{adhwaa.alchaab, a.younis, pompili\}@rutgers.edu
}\vspace{-.1in}}
\maketitle

%%%%%%%%%%%%activate this if needed%%%%
%\linespread{.87}
%%%%%%%%%%%%%%%%%%

%\thispagestyle{empty} \pagestyle{plain} \pagenumbering{arabic}

\begin{abstract} 
Open Radio Access Networks (O-RANs) bring intelligence to disaggregated RANs through virtualized platforms and open interfaces, but the stringent latency budgets of the RAN Intelligent Controller (RIC) and heterogeneous traffic conditions make real-time deployment costly. Centralized learning is ill-suited to the RIC due to its computational burden, timing constraints, and communication overhead, motivating Federated Learning (FL) as a distributed alternative. We formulate FL in O-RAN as a Mixed-Integer Nonlinear Program (MINLP) that minimizes communication cost and learning time under heterogeneity in User Equipment (UE) compute, traffic, and channel conditions, and decompose it into a UE Assignment (UA) subproblem and a Computation and Bandwidth Allocation (CBA) subproblem. We then propose FedRIC, a heterogeneity-aware FL framework combining compressed gradient descent, slice-aware resource allocation, and multilevel hierarchical aggregation. Experiments show that FedRIC outperforms FedCOM, FedAvg, and QFedAvg in convergence rate, communication cost, and latency.
\end{abstract}

%%%%%%%%%%%%%%%%%%%%%%%%%%%%%%%%%%%%%
\section {Introduction}
\textbf{Background:} The growing demand for reliable, high-capacity wireless connectivity has driven mobile operators to redesign network management around software-centricity, virtualization, adaptability, intelligence, and energy efficiency~\cite{10613612}. While the cellular core has embraced these paradigms through SDN and NFV, the Radio Access Network (RAN) has remained largely static. The O-RAN Alliance addresses this gap by promoting disaggregated, cloud-native RAN architectures whose open interfaces and RAN Intelligent Controller (RIC) support emerging use cases such as QoE optimization, SLA assurance, uRLLC, eMBB, and real-time analytics~\cite{sufg2025oran}.

RAN cloudification, however, introduces its own challenges: distributing microservices across a cloud environment inflates compute cost, and running control logic as cloud-native services can fail to meet the tight guarantees needed to close control loops at real-time or near-real-time scales.
%~\cite{10329954}. 
Centralized ML compounds these issues, since it (i) transfers large volumes of locally collected data to a central server, (ii) struggles to satisfy latency and scalability requirements during aggregation, and (iii) exposes raw data, raising privacy and security concerns. These limitations motivate distributed learning paradigms that respect O-RAN's real-time budgets while preserving data locality. Federated Learning (FL) has emerged as one such paradigm for O-RAN intelligence~\cite{9973714}, allowing UEs, gNBs, and the RIC to jointly train a global model without sharing raw data.
\begin{figure}[t]
  \centering
  \begin{subfigure}[t]{0.49\columnwidth}
    \centering
    \includegraphics[width=\linewidth,height=1.15in]{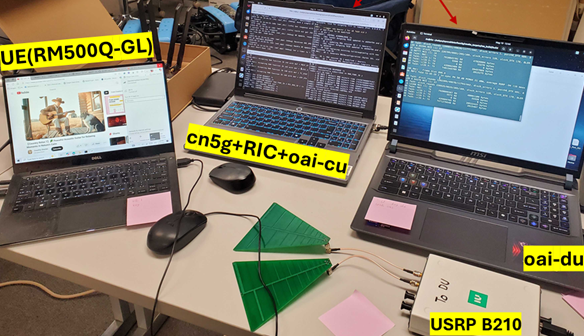}
    \caption{}
    \label{fig:net:arch}
  \end{subfigure}
  \hfill
  \begin{subfigure}[t]{0.49\columnwidth}
    \centering
    \includegraphics[width=\linewidth,height=1.25in]{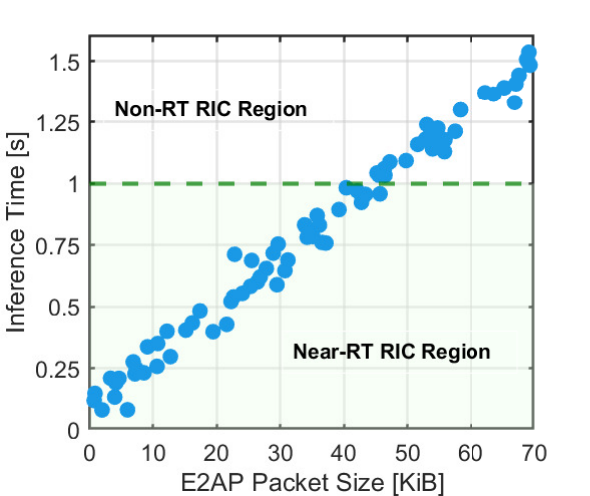}
    \caption{}
    \label{fig:net:other}
  \end{subfigure}

  \vspace{-0.08in}
  \caption{(a) Experimental O-RAN testbed setup; and (b) Inference time vs.\ packet size.}
  \label{fig:two_in_one}
  \vspace{-0.15in}
\end{figure}
\textbf{Motivation:}
While timing constraints of virtualized RANs have been extensively studied in the user plane~\cite{garcia2021nuberu,10723502},
%8005353},
providing the same guarantees in the control plane—particularly for RIC control loops and decision-making—remains an open challenge. Control decisions that arrive after they become obsolete render the closed loop ineffective. To quantify these constraints under realistic conditions, we profiled an O-RAN-compliant Near-RT RIC on the SDR testbed shown in Fig.~\ref{fig:two_in_one}(a), built from the OAI 5G stack~\cite{oai} (CU/DU/core) with USRP B210 radios, FlexRIC hosting a KPI Monitor xApp, and OAI/Quectel RM500Q-GL UEs; further deployment details are reported in our prior work~\cite{11027044}. We generated downlink traffic with iPerf and captured E2AP packets with Wireshark to measure the E2AP processing latency, i.e., the time from dequeuing input data at the gNB to delivery at the xApp. Figure~\ref{fig:two_in_one}(b) reports this latency against E2AP packet size; the shaded region marks the $1~\rm{s}$ Near-RT RIC deadline defined by the O-RAN Alliance, above which control loops fall into the Non-RT regime. Once packets exceed roughly $40~\rm{KiB}$, parsing, validation, and extraction overhead push processing past the $1~\rm{s}$ threshold and the xApp misses its deadline, rendering the control loop ineffective. We conclude that \emph{poorly managed communication in O-RAN environments can easily lead to such deadline violations.} These findings expose a critical gap for deploying FL in O-RAN: unlike cloud-based FL, learning rounds need to meet hard deadlines to keep control loops effective, and heterogeneity in UE compute, channel conditions, and traffic loads.
%—largely unaddressed in prior work—creates variable training times and communication costs that make this requirement even harder to satisfy.
%%%%%%%%%%%%%%%%%%%

\textbf{Related Works.}
A range of FL algorithms address data heterogeneity, network instability, and computational constraints in distributed learning. FedDist~\cite{9439129} evolves model architecture by adding diverging neurons identified through Euclidean distance, GFL~\cite{10099110} leverages differentially private GANs to synthesize balanced IID data and mitigate non-IID skew, and OPTIMIST~\cite{11038656} performs opportunistic federated transfer learning over transient device encounters using partitioned subnetworks. While effective in their target settings, none of these approaches account for the hard timing budgets imposed by RIC control loops.

Communication-efficient FL has also been studied as a way to improve convergence under limited uplink capacity~\cite{9973714,bernstein2018signsgd,alistarh2017qsgd}, most commonly by quantizing local updates before transmission. Bernstein et al.~\cite{bernstein2018signsgd} apply a 1-bit sign quantizer to each entry of the local update, while Alistarh et al.~\cite{alistarh2017qsgd} adopt stochastic quantization that preserves the statistical properties of the updates. Related lines of work examine FL over wireless IoT networks~\cite{9712615}, fronthaul-aware FL in NG-RANs~\cite{9973714,11168951}, and federated analytics for 6G~\cite{10401940}. However, these solutions optimize for eventual convergence rather than deadline-constrained learning, leaving a gap for O-RAN control loops that operate on millisecond-to-second timeframes and contend with heterogeneous UE compute, channels, and traffic.

\textbf{Contributions:} The primary contributions are: (1) formulating FL optimization in O-RAN as a Mixed-Integer Nonlinear Program (MINLP) that jointly minimizes communication cost and training time under RIC deadline constraints; (2) decomposing the joint problem into UE Assignment (UA) and Computation and Bandwidth Allocation (CBA) subproblems; (3) proposing FedRIC and FedRICE, FL frameworks integrating compressed gradient descent, slice-aware resource allocation, and hierarchical aggregation; and (4) demonstrating through experiments on real O-RAN datasets that our approach satisfies hard RIC timing constraints while achieving faster convergence and reduced overhead. The remainder of this paper is organized as follows. Sect.~\ref{sec:SM} provides an overview of the O-RAN architecture, the system model, and the problem formulation. Sect.~\ref{sec:PS} describes the proposed FL solution in detail. Sect.~\ref{sec:PE} presents the performance evaluation results. Finally, Sect.~\ref{sec:CO} concludes the paper.

%%%%%%%%%%%%%%%%%%%%%%%%%%%%%%%%%%%%%
\section{System Model}\label{sec:SM}
As shown in Fig.~\ref{fig:NET}, the term \emph{E2 node} denotes the logical node comprising the CU, DU, and RU as specified in 3GPP TS 38.401~\cite{3gpp_ts_38_401}, while \emph{RIC node} refers jointly to the Non-RT RIC (hosting rApps) and the Near-RT RIC (hosting xApps). We consider a single-cell network with a logical E2 node serving a set of stationary UEs $u \in \{1,\ldots,U\}$. The system bandwidth is divided into $K$ orthogonal OFDM channels, with inter-channel interference mitigated using standard methods~\cite{zhao2001intercarrier}. The channel between each transmitter–receiver pair is subject to small-scale fading, path loss, and shadowing.

\subsection{Learning Model}
Each UE $u \in \mathcal{U}$ holds a local dataset $D_{u} = \{(\mathbf{x}_{ud}, y_{ud})\}_{d=1}^{|D_{u}|}$ with $\mathbf{x}_{ud} \in \mathbb{R}^f$ and $y_{ud} \in \mathbb{R}$. We assume the data are non-i.i.d. and disjoint, i.e., $D_u \cap D_{u'} = \emptyset$ for $u \neq u'$. Let $\mathbf{v}$ denote the global model parameters and $\mathbf{v}_u$ the local parameters at UE $u$. The FL objective is,
\begin{subequations}
\begin{align}
\min_{\mathbf{v}} \quad & \frac{1}{|\mathcal{D}|} \sum_{u=1}^{U} \sum_{d=1}^{|D_u|} \varphi_u(\mathbf{v}_u, \mathbf{x}_{ud}, y_{ud}) \label{eq:objective} \\
\text{s.t.} \quad & \|\mathbf{v}_u - \mathbf{v}\| = 0, \quad \forall u \in \mathcal{U},\label{eq:constraint}
\end{align}
\end{subequations}
where $|\mathcal{D}| = \sum_{u=1}^U |D_u|$ and $\varphi_u(\cdot)$ is UE $u$'s loss (e.g., MSE or cross-entropy). Constraint~\eqref{eq:constraint} synchronizes local and global parameters after each round.

\begin{figure}[t]
\centering
\includegraphics[width=0.38\textwidth]{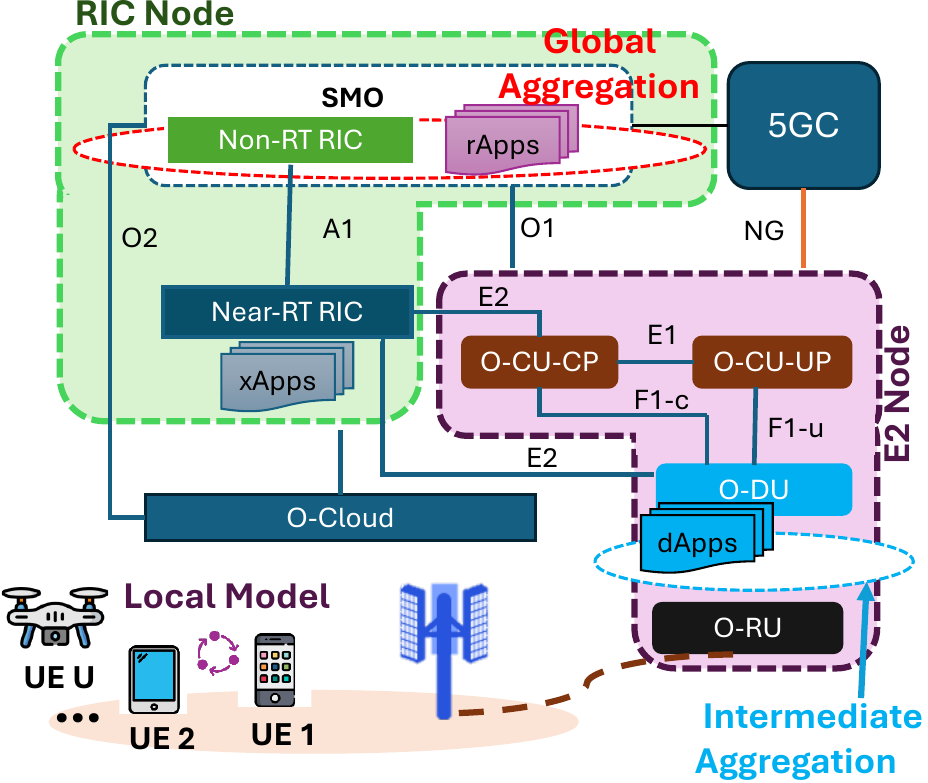}
\vspace{-0.1in}
\caption{Proposed O-RAN architecture with FL.}
\label{fig:NET}
\vspace{-0.2in}
\end{figure}

\subsection{FL Model Accuracy}
We define local and global accuracy targets at the UE and RIC node, respectively. Each UE's local target $\xi^{\rm{l}}\in(0,1)$ requires that the gradient norm shrink by at least a factor of $\xi^{\rm{l}}$ between successive global iterations,
\begin{equation}\label{eq:ac}
||\nabla \varphi_u^{[i]}(\mathbf{v}_u^{[i]})|| \leq \xi^{\rm{l}} ||\nabla \varphi_u^{[i]}(\mathbf{v}_u^{[i-1]})||, \quad \forall u \in \m{U}.
\end{equation}
The global target $\xi^{\rm{g}}$ (with $0 < \xi^{\rm{g}} \ll 1$) is reached when the global loss change satisfies,
\begin{equation}\label{eq:I}
|\varphi(\mathbf{v}^{[i]}) - \varphi(\mathbf{v})| \leq \xi^{\rm{g}}, \quad \forall i \geq I,
\end{equation}
where $I$ is a sufficient number of global iterations. Under standard convexity, $\rho$-Lipschitz, and $\beta$-smoothness assumptions on $\varphi(\mathbf{v})$, the number of global iterations needed to reach $\xi^{\rm{g}}$ given $\xi^{\rm{l}}$ is bounded by,
%\begin{equation}\label{eq:o}
$\Upsilon(\xi^{\rm{g}}, \xi^{\rm{l}}) = O\left(\frac{\ln(1/\xi^{\rm{g}})}{1 - \xi^{\rm{l}}}\right),$
%\end{equation}
which links accuracy targets to training rounds and underpins the cost and latency models below.

\subsection{Resource Model}
In each global iteration $i$, a binary variable $\mu^{[i]}_u \in \{0,1\}$ indicates whether UE $u$ participates, with selection vector $\bm{\mu}^{[i]} = \{\mu^{[i]}_1,\ldots,\mu^{[i]}_U\}$. Inactive UEs ($\mu^{[i]}_u = 0$) receive no bandwidth ($w^{[i]}_u = 0$); active UEs are guaranteed at least $w_{\min} \leq 1/U$ of the total bandwidth $W$. With $\varepsilon$ as the cost per unit bandwidth, the total communication cost over $\Upsilon^{\rm{g}}$ global rounds is,
%\begin{equation}\label{eq:phi}
$\Phi^{\text{m}} = \Upsilon^{\text{g}} \sum\nolimits_{u\in\m{U}} \mu^{[i]}_{u} w^{[i]}_u W \epsilon.$
%\end{equation}
Letting $\tau$ be the cost per unit processing time, $\kappa_u$ the CPU rate, and $f_u$ the cycles per bit, the total computation cost is,
%\begin{equation}\label{eq:p}
$\Phi^{\rm{p}} = \sum\nolimits_{t\in\m{T}} \mu^t_u \tau D_u f_u/\kappa_u,$
%\end{equation}
where $|\mathcal{T}| = \Upsilon^{\rm{g}}$.

\subsection{Latency Model}
We adopt a synchronous protocol in which each global round waits for all selected UEs to complete local training and upload their updates, so the round is paced by the slowest UE. With $T^{\rm{p}}_u$ the per-iteration compute time, $I^{\rm{l}}$ local iterations, update size $|D_u^{\rm{up}}|$, and upload time $T^{\rm{m}}_u = |D_u^{\rm{up}}|/(w^{[i]}_u W)$, the per-round duration at UE $u$ is,
%\begin{equation}
$T^{\rm{g}}_u = I^{\rm{l}} T^{\rm{p}}_u + T^{\rm{m}}_u,\quad u \in \m{U}.$
%\end{equation}
We omit downlink delay since the downlink rate is much higher than the uplink and the global model is much smaller than the aggregate local data, as in~\cite{10328982}. With $T_{\max} = \max_{u \in \mathcal{U}} T^{\rm{g}}_u$ and $\Upsilon^{\rm{g}}$ rounds to reach $\xi^{\rm{g}}$, the total training time is $T_{\text{total}} = \Upsilon^{\rm{g}} T_{\max}$. Combining resource cost and latency through a weighting parameter $\delta$ yields the per-iteration cost, $\Omega^{[i]} = \Phi^{\rm{m}} + \Phi^{\rm{P}} + \delta\,\Upsilon^{\rm{g}} \max\{T^{\rm{g}}_u; \forall u \in \mathcal{U}\}.$

\subsection{Problem Formulation}
The design goal is to minimize $\Omega^{[i]}$ subject to the system constraints,
\begin{subequations}\label{eq:cost}
\begin{align}
\min \Omega^{[i]} &= \Phi^{\rm{m}} + \Phi^{\rm{P}} + \delta\,\Upsilon^{\rm{g}} \max\{T^{\rm{g}}_u\} \label{eq:cost_a} \\
\text{s.t.} \quad &\sum\nolimits_{u \in \m{U}} w^{[i]}_u = 1, \quad \forall u \in U \label{eq:cost_b} \\
&w^{[i]}_u \geq w_{\min} \text{ if } \mu^{[i]}_u = 1, \;\; w^{[i]}_u = 0 \text{ if } \mu^{[i]}_u = 0 \label{eq:cost_cd} \\
&\sum\nolimits_{u\in\m{U}} w^{[i]}_u \mu^{[i]}_u W \leq W \label{eq:cost_e} \\
&\sum\nolimits_{u \in \m{U}} |D_u| f_u/\kappa_u \leq F_{\max} \label{cost_f}\\
&\eqref{eq:ac},\eqref{eq:I}.\label{cost_g}
\end{align}
\end{subequations}
Constraints~\eqref{eq:cost_b}–\eqref{eq:cost_e} govern bandwidth allocation: total allocations are normalized, active UEs receive at least $w_{\min}$, inactive UEs receive none, and total usage stays within capacity $W$. Constraint~\eqref{cost_f} caps aggregate compute demand at $F_{\max}$, and \eqref{cost_g} enforces the local and global accuracy targets.

%%%%%%%%%%%%%%
\section{Proposed FedRIC Solution}\label{sec:PS}
The optimization problem in~\eqref{eq:cost} represents a Mixed Integer Nonlinear Program (MINLP), which is NP-hard to solve due to three primary challenges. \emph{First}, the binary variables $\boldsymbol{\mu} = \{\mu_u^{[i]}\}$ create $2^{|\mathcal{U}|}$ possible UE selection combinations per round, resulting in exponential complexity across $\Upsilon_g$ global iterations. \emph{Second}, the objective function exhibits non-convex coupling between discrete selection variables $\mu_u^{[i]}$ and continuous allocation variables $(w_u^{[i]}, \theta_u^{[i]})$, preventing the use of standard convex optimization techniques. \emph{Third}, the strict deadline requirements necessitate solutions computable within milliseconds, while exact MINLP solvers typically require minutes to hours for convergence. To address these challenges, we consider two deployment variants of the framework. \emph{FedRIC} is a two-layer hierarchical design in which each selected UE transmits its compressed gradient to its serving DU, local updates are aggregated at the DU with momentum smoothing, and the result is forwarded to the CU at the E2 node every $\Delta_{\text{CU}}$ rounds for global aggregation. \emph{FedRICE} extends this design with a third aggregation layer at the O-Cloud level~\cite{oran-r004-v14-2025}.

\subsection{Problem Decomposition}
\subsubsection{UE Assignment (UA, $\mathcal{P}1$ subproblem)}
In this stage, we determine the binary variables $\boldsymbol{\mu} = \{\mu_u^{[i]}\}$ indicating which UEs participate in each round. Our objective is to maximize the number of participating UEs in every global iteration, thereby incorporating as much data as possible and accelerating convergence, while ensuring that each round's duration remains within the deadline. The resulting subproblem can be formulated as a 0–1 integer program,
\begin{equation}\label{eq:p1}
\begin{aligned}
\mathcal{P}1: \quad & \min_{\boldsymbol{\mu} \in \{0,1\}} \ \Upsilon^{\rm{g}} \sum_{i \in \mathcal{I}} \sum_{u \in \mathcal{U}} \mu_u^{[i]} \left( w_u^{[i]} W \varepsilon + \frac{D_u f_u}{\kappa_u} \tau \right) \\
\text{s.t.} \quad & \mu_u^{[i]} \in \{0,1\}, \quad \forall u \in \mathcal{U},\ \forall i \in \mathcal{I}
\end{aligned}
\end{equation}
%%%%%%%%%%

The objective in~\eqref{eq:p1} represents total time cost combining transmission and computation terms for selected UEs. Rather than directly minimizing this cost, we use a slice-aware heuristic to maximize participating UEs per round while respecting deadline constraints, ensuring balanced representation across traffic slices (eMBB, uRLLC, mMTC). Let $i_{\text{ite}}$ be the maximum round time, $i_1, i_a$ be overhead times, and $i_k^m$ be the estimated time for UEs from slice $m$. The per-round UA problem becomes,
%%%%%%%%%%
%%%
\begin{equation}\label{eq:max}
\max_{M_u \subseteq \mathcal{U}} |M_u| \quad \text{s.t.} \quad i_1 + i_a + \frac{1}{2} \left( i_{k-1}^m + \zeta i_k^m \right) \leq i_{\text{ite}}
\end{equation}
where $M_u$ is the set of selected UEs and $\zeta > 0$ weights the incoming slice's time contribution. The constraint in~\eqref{eq:max} ensures that the weighted incremental time $\frac{1}{2}(i_{k-1}^m + \zeta i_k^m)$ for adding a UE stays within the remaining deadline. This is a knapsack-style problem under a time budget and is NP-hard, so we approximate the solution with the latency-aware heuristic in Algorithm~\ref{alg:ua}, which iteratively admits the UE with the smallest incremental time cost, dynamically boosts underrepresented slices (eMBB, uRLLC, mMTC) to preserve diversity, and terminates once any further addition would violate $i_{\text{ite}}$.
\begin{algorithm}[!t]
\caption{Latency-aware UE assignment}
\label{alg:ua}
\footnotesize
\begin{algorithmic}[1]
\Require UE set $\mathcal{U}$, deadline $i_{\text{ite}}$, overheads $i_1,i_a$
\Ensure Selected UE set $M_u$
\State $M_u \gets \emptyset,\; t_{\text{used}} \gets i_1+i_a$
\State $\mathcal{Q} \gets \mathrm{sort}\!\left(\mathcal{U},\, i_k^m/\mathrm{data\_size}_k\right)$
\While{$\mathcal{Q}\neq\emptyset$}
    \State $u^* \gets \arg\min\limits_{u\in\mathcal{Q}} \frac{1}{2}\!\left(i_{k-1}^{m(u)}+\zeta i_k^{m(u)}\right)$
    \State $t_{\text{inc}} \gets \frac{1}{2}\!\left(i_{k-1}^{m(u^*)}+\zeta i_k^{m(u^*)}\right)$
    \If{$t_{\text{used}}+t_{\text{inc}} > i_{\text{ite}}$} \State \textbf{break} \EndIf
    \State $M_u \gets M_u \cup \{u^*\},\; t_{\text{used}} \gets t_{\text{used}}+t_{\text{inc}}$
    \State $\mathcal{Q} \gets \mathcal{Q}\setminus\{u^*\}$; boost underrepresented slices by $\alpha$
\EndWhile
\State Enforce minimum slice representation if feasible
\Return $M_u$

\end{algorithmic}
\end{algorithm}

\subsubsection{Computation and Bandwidth Allocation (CBA, $\mathcal{P}2$ subproblem)} 
Given the set of UEs selected in the UA stage, we now allocate communication bandwidth and computational resources for the current round. Each UE $u \in M_u$ receives a bandwidth share $w_u^{[i]} \in [0,1]$ and a computational allocation $\theta_u$ chosen to minimize the round's duration. Substituting the fixed selection variables $\boldsymbol{\mu}^{[i]}$ from the UA stage into the original problem~\eqref{eq:cost} yields the reduced optimization,
\begin{equation}\label{eq:p2}
\begin{aligned}
\mathcal{P}2: \quad & \min_{w_u^{[i]}, \theta_u} \max_{u \in M_u} \left\{ I^l \frac{D_u f_u}{\kappa_u \theta_u} + \frac{|D_u^{\text{up}}|}{w_u^{[i]} W} \right\} \\
\text{s.t.} \quad & \sum_{u \in M_u} w_u^{[i]} = 1,\; w_u^{[i]} \geq w_{\min},\; \sum_{u \in M_u} \theta_u \leq \Theta_{\max},
\end{aligned}
\end{equation}
where $I^l$ is the number of local iterations, $|D_u^{\text{up}}|$ is the size of UE $u$'s compressed update, $W$ is the total available bandwidth, and $\Theta_{\max}$ is the total computational budget. The resulting CBA subproblem is NP-hard due to the min-max structure and the coupling between bandwidth and computation time in the objective. Given the millisecond-scale decision requirements of O-RAN control loops, we employ a fast proportional allocation heuristic that balances completion times across UEs. Computational resources are allocated proportionally to processing requirements,
$\theta_u = \Theta_{\max} \cdot \frac{D_u f_u/\kappa_u}{\sum_{u' \in M_u} (D_{u'} f_{u'}/\kappa_{u'})},$
while bandwidth is distributed proportionally to upload sizes with a minimum guarantee,
$w_u^{[i]} = \max\left\{w_{\min},\, \frac{|D_u^{\text{up}}|}{\sum_{u' \in M_u} |D_{u'}^{\text{up}}|}\right\},$
followed by normalization so that $\sum_{u \in M_u} w_u^{[i]} = 1$. This greedy procedure operates in $O(|M_u|)$ time, achieving near-balanced completion times without the overhead of exact optimization. The RIC node then conveys the allocated $w_u^{[i]}$ and $\theta_u$ to each selected UE $u \in M_u$, and the round proceeds with local training and model upload under these allocations.

%%%%%%%%%%
\subsection{Local Training with Compressed AGD}

Once the UA and CBA phases determine the selected UEs and their allocated resources, each chosen UE $u \in M_u$ performs local model training using its assigned bandwidth $w_u^{[i]}$ and computational allocation $\theta_u$. To meet the strict deadline constraints while maintaining convergence quality, we develop a compressed Accelerated Gradient Descent (AGD) approach, inspired by~\cite{8812696}, that combines momentum-based acceleration with randomized compression. Each selected UE $u$ computes its local gradient and applies a randomized compression operator $\mathcal{C}(\cdot)$ before transmission, $\mathbf{g}_u^{[i]} = \nabla \varphi_u(\mathbf{v}_u^{[i]}),
\mathbf{c}_u^{[i]} = \mathcal{C}(\mathbf{g}_u^{[i]})$, where $\mathbf{g}_u^{[i]}$ is the full gradient and $\mathbf{c}_u^{[i]}$ is the compressed version. The compression operator satisfies $\mathbb{E}[\mathcal{C}(\mathbf{x})] = \mathbf{x}$ (unbiased) and $\mathbb{E}[\|\mathcal{C}(\mathbf{x}) - \mathbf{x}\|^2] \leq \theta \|\mathbf{x}\|^2$ for some compression parameter $\theta \geq 0$. We instantiate $\mathcal{C}$
as a stochastic quantizer~\cite{alistarh2017qsgd} with a per-round bit budget that decreases from 16 down to 2 bits over successive global rounds, which lets the variance bound $\theta$ tighten as the model approaches convergence.

\textbf{Momentum-Based Updates:} To accelerate convergence under communication constraints, each UE maintains a momentum term that incorporates compressed gradients,
\begin{equation}\label{eq:momentum_update}
\begin{aligned}
\mathbf{z}_u^{[i+1]} &= \beta \mathbf{z}_u^{[i]} + \mathbf{c}_u^{[i]},
\mathbf{v}_u^{[i+1]} &= \mathbf{v}_u^{[i]} - \eta \mathbf{z}_u^{[i+1]}
\end{aligned}
\end{equation}
where $\beta \in [0,1)$ is the momentum coefficient and $\eta > 0$ is the learning rate. The momentum term $\mathbf{z}_u^{[i]}$ accumulates historical gradient information, enabling faster convergence even with compressed communications. The local training process respects the resource allocations from the CBA phase. Each UE performs $I^l$ local iterations within its computational budget $\theta_u$, ensuring that,
$I^l \cdot \frac{|\mathcal{D}_u| f_u}{\kappa_u} \leq \theta_u  i_{\text{ite}}$,
where $|\mathcal{D}_u|$ is the local dataset size, $f_u$ is the computational complexity per sample, $\kappa_u$ is the processing power, and $i_{\text{ite}}$ is the allocated time budget. The combination of compression and momentum in FedRIC maintains convergence guarantees under deadline constraints.

\begin{theorem}\label{thm:convergence}
Let $\mathcal{C}$ be an unbiased compression operator with variance bound $\theta$. Under standard $S$-smoothness and $\varrho$-strong convexity assumptions on $\varphi$, FedRIC with compressed AGD satisfies,
$\mathbb{E}[\varphi(\mathbf{v}^{[i]})-\varphi(\mathbf{v}^*)] \leq \frac{2S\|\mathbf{v}_0 - \mathbf{v}^*\|^2}{i(\varrho - S\eta)} + \frac{2\eta\psi^2}{\varrho} + \frac{\theta}{2}\|\mathbf{v}^{[i]} - \mathbf{v}^*\|^2$,
where $\mathbf{v}^*$ is the optimum and $\psi^2$ bounds the stochastic gradient variance.
\end{theorem}
\begin{proof}
The bound follows from the AGD-with-compression analysis in~\cite{8812696}, adapted to the momentum update~\eqref{eq:momentum_update}: unbiasedness of $\mathcal{C}$ preserves the expected gradient direction, the variance bound contributes the $\frac{\theta}{2}\|\mathbf{v}^{[i]} - \mathbf{v}^*\|^2$ term, and the momentum factor yields the $i(\varrho - S\eta)$ denominator.
\end{proof}

%%%%%%%%%%%%
\textbf{Hierarchical Aggregation and Workflow.} FedRIC aggregates compressed updates in two stages: at each DU, co-located UE updates are combined via a data-size-weighted average and smoothed with DU-level momentum, then forwarded to the CU every $\Delta_{\text{CU}}$ rounds for global update, $\mathbf{v}^{[i+1]} = \sum_{u \in M_u} \frac{|D_u|}{\sum_{j \in M_u}|D_j|}\, \mathbf{v}_u^{[i+1]}.$ FedRICE inserts an extra O-Cloud edge-cluster step before the DU stage, applying the same data-weighted rule so that only cluster-level aggregates traverse the fronthaul, at the cost of one additional hop per round. Each global round therefore executes: (i)~latency-aware UE assignment via Algorithm~\ref{alg:ua}; (ii)~proportional CBA allocation of $(w_u^{[i]}, \theta_u)$; (iii)~$I^l$ local AGD steps per selected UE with compressed gradients, where $I^l$ is the largest integer satisfying $I^l \cdot |\mathcal{D}_u| f_u/\kappa_u \leq \theta_u\, i_{\text{ite}}$; and (iv)~hierarchical aggregation as above.

\section{Performance Evaluation}\label{sec:PE}
All experiments were conducted on a workstation with an Intel Core i9-14900K CPU (24 cores, 32 threads, up to 6~GHz). We implement a two-layer hierarchical aggregation strategy that dynamically selects UEs based on their local training times and resource allocations. Uplink bandwidth per round is constrained between 10 and 1000 units, and model updates are quantized from 16 bits down to 2 bits over successive rounds to simulate compression under varying conditions.

%%%%%%%%%

%%%%%%%%%%

%%%%%%%%%
 
  %\caption{ (a) ROC curves for four sequence models; (b) Accuracy vs. training time on the ColO-RAN dataset.}\label{fig:FP}\vspace{-0.2in} \end{figure}

%%%%%%%%

%%%%%%%%%%%%%%%%%%%%%%%%%%%%%
\begin{figure*}[!t]
\centering
\begin{minipage}{0.2\textwidth}
    \centering
    \includegraphics[width=\linewidth]{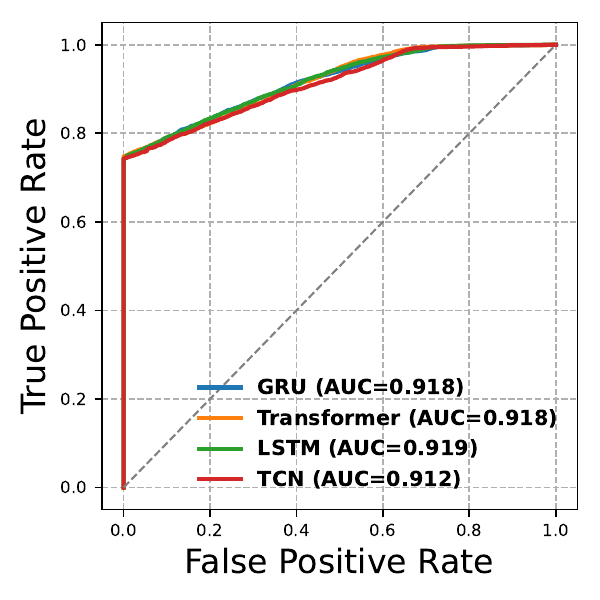}
    \vspace{-5mm}
    
    {\footnotesize (a)}
\end{minipage}\hfill
\begin{minipage}{0.245\textwidth}
    \centering
    \includegraphics[width=\linewidth]{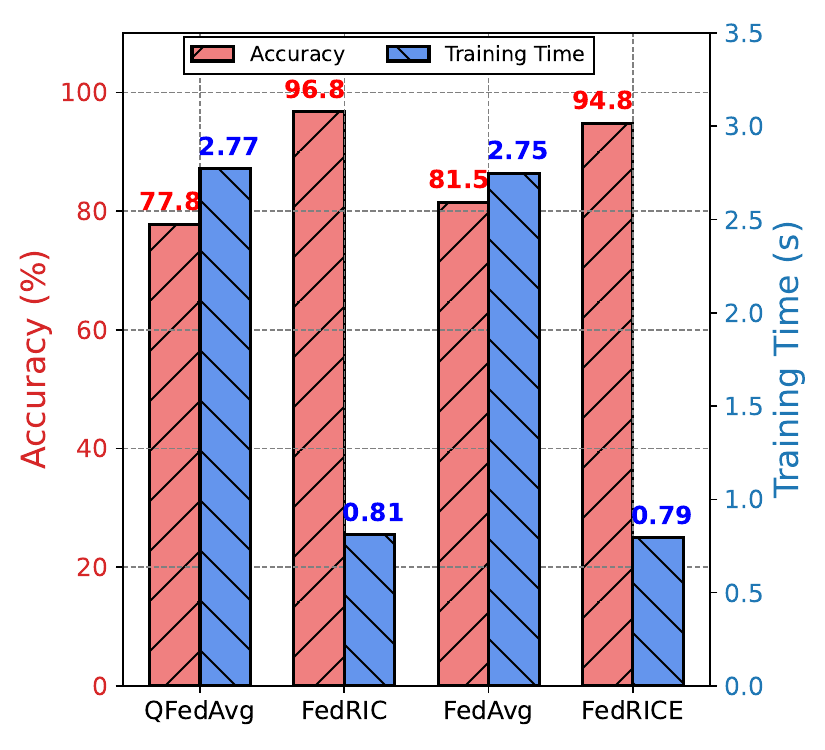}
    \vspace{-5mm}
    
    {\footnotesize (b)}
\end{minipage}\hfill
\begin{minipage}{0.245\textwidth}
    \centering
    \includegraphics[width=\linewidth]{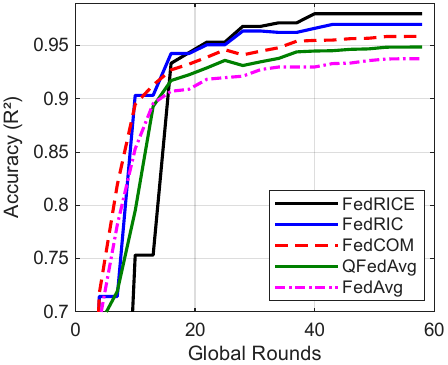}
    \vspace{-5mm}
    
    {\footnotesize (c)}
\end{minipage}\hfill
\begin{minipage}{0.245\textwidth}
    \centering
    \includegraphics[width=\linewidth]{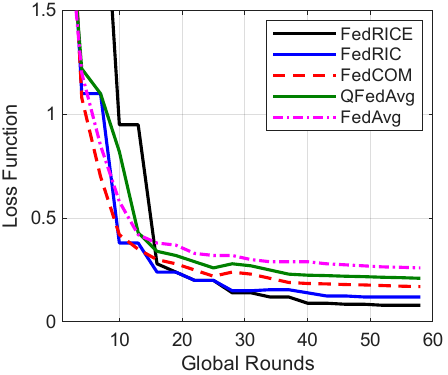}
    \vspace{-5mm}
    
    {\footnotesize (d)}
\end{minipage}

\vspace{-2mm}
\caption{Performance comparison on the ColO-RAN dataset: (a) ROC curves for four sequence models; (b) accuracy versus training time; (c) accuracy versus global rounds; and (d) loss versus communicated bits.}
\label{fig:dl}
\vspace{-2mm}
\end{figure*}

%%%%%%
\begin{figure*}[!t]
\centering
\begin{minipage}{0.24\textwidth}
    \centering
    \includegraphics[width=\linewidth]{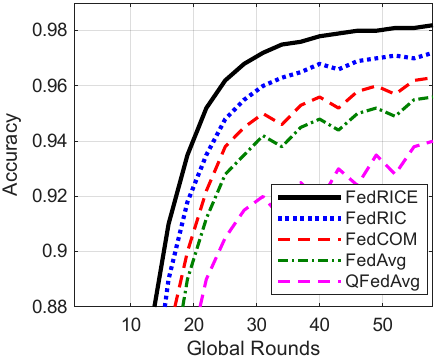}
    \vspace{-5mm}
    
    {\footnotesize (a)}
\end{minipage}\hfill
\begin{minipage}{0.24\textwidth}
    \centering
    \includegraphics[width=\linewidth]{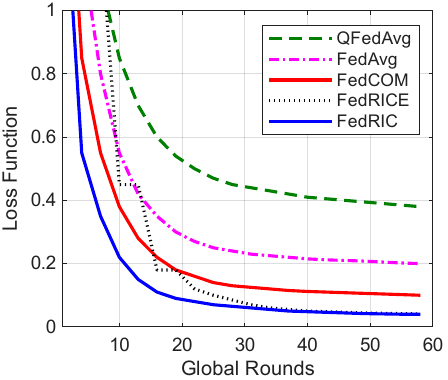}
    \vspace{-5mm}
    
    {\footnotesize (b)}
\end{minipage}\hfill
\begin{minipage}{0.23\textwidth}
    \centering
    \includegraphics[width=\linewidth]{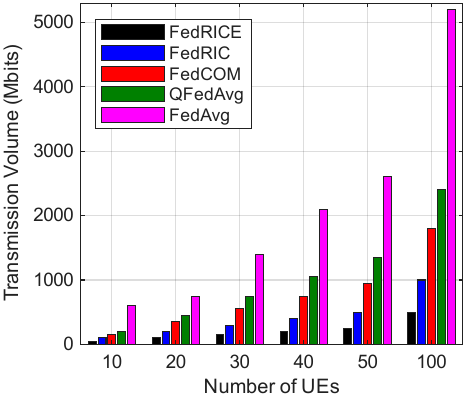}
    \vspace{-5mm}
    
    {\footnotesize (c)}
\end{minipage}\hfill
\begin{minipage}{0.21\textwidth}
    \centering
    \includegraphics[width=\linewidth]{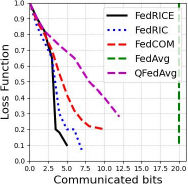}
    \vspace{-5mm}
    {\footnotesize (d)}
\end{minipage}
\vspace{0.1mm}
\caption{Comparison on the Telecom Italia cellular dataset:  (a)  accuracy  versus global rounds; (b) loss versus global rounds; (c) transmission volume versus number of UEs; and (d) effect of the compression operator.}
\label{fig:LE}
\vspace{-5mm}
\end{figure*}

%%%%%%%%
\begin{table}[htbp]
\caption{Model comparison on key metrics.}
\vspace{-0.1in}
  \centering
  \small
  \begin{tabular}{@{}lcccc@{}}
    \toprule
    \textbf{Model}   & \textbf{Accuracy} & \textbf{AUC} & \textbf{Train (s)} & \textbf{Infer (ms)} \\
    \midrule
    GRU              & 0.9584            & 0.918        & 21.240             & 3.54                \\
    Transformer      & 0.9625            & 0.918        & 63.737             & 12.22               \\
    LSTM             & 0.9594            & 0.919        & 25.672             & 4.02                \\
    TCN              & 0.9579            & 0.912        & 59.603             & 1.04                \\
    \bottomrule
  \end{tabular}
  
  \label{tab:p}
  \vspace{-0.2in}
\end{table}
%%%%%%%%%%%%%%%
%%%%%%%%%%%%%%%%%%

\subsection{Experimental Setting}
\textbf{FL Tasks.} We first evaluate a supervised KPI‐violation detector on the ColO-RAN dataset~\cite{polese2022coloran}, which comprises 18\,329 one-second metric snapshots from seven O-RAN gNBs under three scheduling policies. Each snapshot is labeled according to whether the 1~s Near-RT deadline was missed, and the full dataset is uniformly partitioned across 50 simulated RIC nodes. For a centralized benchmark, we train four sequence models: a single-layer GRU with 64 hidden units; 
%(tanh activation, Adam optimizer)
a two-layer Transformer encoder;
%($d_{\mathrm{model}}=64$, four heads, Adam)
a single-layer LSTM with 64 units; 
%(sigmoid activation, Adam)
and a TCN comprising a 16-channel Conv1d, ReLU, global average pooling, and a dense output layer. All models are trained for $20$ epochs with a batch size of $64$ using binary cross-entropy loss with logits applied to raw output scores. Based on Fig.~\ref{fig:dl}(a) and Table~\ref{tab:p} measurements, we adopt the LSTM as our local FL backbone since it offers the best balance, $95.94\%$ accuracy AUC $0.919$ with moderate training time $25.7~\rm{s}$ and inference latency $4.02\rm{ms}$.

\textbf{Hierarchical Aggregation.}
%We evaluate our framework (Algorithm~\ref{alg:fedric}) under two deployment schemes: FedRIC, our two-layer hierarchical design in which UEs compress their gradients and send them to the E2 node; and FedRICE, which extends FedRIC by adding an additional O-RAN edge layer (analogous to the O-Cloud in the O-RAN architecture) that participates alongside the E2 node for distributed training. 
We evaluate our framework under two deployment schemes: FedRIC, our two-layer hierarchical design in which UEs compress their gradients and send them to the E2 node; and FedRICE, which extends FedRIC by adding an intermediate aggregation layer at the O-Cloud level~\cite{oran-r004-v14-2025}, creating a three-layer hierarchy where local updates are first aggregated at edge clusters before being forwarded to the E2 node for final global aggregation.
In FedRIC, local updates are first aggregated at the DUs with momentum smoothing before being forwarded to the CU every \(\Delta_{\text{CU}}\) rounds. FedRICE further inserts an edge-level aggregation step in each round, where UEs within the same edge group average their models before DU aggregation. For comparison, we benchmark against three established FL baselines: FedAvg~\cite{mcmahan2017communication}, the standard FL algorithm that uses a fixed subset of UEs per round without compression or client selection—serving as a lower-bound reference; QFedAvg~\cite{li2019qffl}, a fairness-aware extension of FedAvg that introduces a parameter \(q\) to re-weight UE contributions; and FedCom~\cite{haddadpour2021federated}, which is a communication-efficient FL baseline that
reduces uplink overhead by compressing UE model updates before aggregation. %which adopts a UE clustering algorithm to merge UE models based on data similarity in terms of accuracy, fronthaul traffic volume per training round, and UE perceived latency.

\subsection{Experiment Results}

 %All computationally intensive FL tasks—namely the joint optimization of UE selection and resource allocation—are first carried out offline in the Service Management and Orchestration (SMO) layer of the Non-RT RIC by a dedicated rApp. Specifically, the rApp ingests historical KPI traces and executes both the UE-Assignment routine (Algorithm~\ref{alg:ua}) and the CBA subproblem within the FedRIC framework (Algorithm~\ref{alg:fedric}) in simulation, tuning the CU-aggregation interval \(\Delta_{\mathrm{CU}}\), the per-round quantization schedule \(\{b_r\}\), and the per-UE latency budgets \(\{t_u\}\) so as to satisfy the global-accuracy target \(\epsilon\) and all per-round timing constraints. Once these parameters are validated, the SMO packages the trained model together with its compression and deadline-aware policies as an xApp and deploys it to the Near-RT RIC for real-time inference and closed-loop control. 

 We model an O-RAN cell with UEs randomly distributed and E2 nodes uniformly placed.  Channel gains follow the distance-dependent path-loss, $L[dB] = 140.7 + 36.7 \log_{10} d[km]$, with fading split into constant large-scale and frequency-selective small-scale components.  Small-scale Rayleigh fading uses Jake’s model as a first-order complex Gauss–Markov process~\cite{10723502}.
 %: $h_u = \rho_0\,h_u + \sqrt{1-\rho_0^2}\,e_u$, where  $\rho_0 = J_0\bigl(2\pi f_d T_0\bigr)$ \(J_0\) is the zero-order Bessel function, \(f_d\) the maximum Doppler frequency, \(T_0\) the re-estimation interval, and \(e_u\sim\mathcal{CN}(0,1)\).  
%%

\textbf{Comparison of Communication Cost.}
In Fig.~\ref{fig:dl}(b), we compare the total training time and model accuracy of the four FL approaches. FedRIC achieves the highest accuracy, $96.76 \%$, with minimal training latency. FedRICE follows closely, attaining $94.82 \% $ accuracy with consistently low latency below $0.9\rm{s}$. Although its accuracy is slightly lower than FedRIC’s, FedRICE still offers a strong balance between performance and efficiency. These results demonstrate the effectiveness of the evaluated approaches in optimizing computational resources without compromising accuracy.

\textbf{Comparison of Accuracy.}
%Fig.~\ref{fig:dl} illustrates the loss and accuracy metrics of the proposed model, showing accuracy progression during training. 
We compare four FL approaches on the ColO-RAN and Telecom Italia cellular datasets. In Fig.~\ref{fig:dl}(c), the accuracy curves on ColO-RAN show that FedRIC and FedRICE consistently outperform the other methods, achieving high accuracy in fewer global rounds. Fig.~\ref{fig:dl}(d) presents the corresponding loss curves, where FedRICE and FedRIC reach lower loss values in fewer rounds. Fig.~\ref{fig:LE}(a) shows the accuracy on the Telecom Italia dataset, where FedRIC and FedRICE achieve higher final accuracy than FedAvg, FedCOM, and QFedAvg within the same number of rounds. Finally, Fig.~\ref{fig:LE}(b) shows the loss curves for the Telecom Italia dataset, confirming that FedRIC and FedRICE achieve loss reduction at a faster rate than the baselines. %These results demonstrate that the proposed algorithms improve prediction accuracy and reduce the number of global rounds required for convergence in resource-limited environments.

\textbf{Comparison of Transmission Volume.}
We evaluate communication efficiency by measuring transmission volumes as the number of UEs varies from $10$ to $100$. Fig.~\ref{fig:LE}(c) shows that FedAvg’s volume increases with UE numbers, reaching $5{,}200$ Mbits at $100$ UEs—underscoring its inefficiency in large-scale deployments. In contrast, the FedRIC and FedRICE frameworks deliver significant improvements: FedRIC reduces the volume to $1{,}000$ Mbits, and FedRICE further lowers it to $500$ Mbits for $100$ UEs, through hierarchical aggregation, latency-driven resource allocation, and compression. FedCOM, while more efficient than QFedAvg and FedAvg, still transmits $1{,}800$ Mbits at $100$ UEs, limiting its scalability.
%These results demonstrate that our proposed algorithms achieve high scalability in communication efficiency, making them well-suited for networks with large user populations.  

\textbf{Compressed Bit Size}.
Fig.~\ref{fig:LE}(d) illustrates the relationship between the loss function value and the size of compressed bits transmitted in each global round. FedAvg, which does not employ any compression technique, consistently sends $20$ bits per round. In contrast, FedRIC, FedRICE, FedCOM, and QFedAvg utilize varying amounts of compressed bits. The results show that both FedRICE and FedRIC achieve lower loss function values with fewer global rounds while transmitting fewer bits compared to FedCOM, QFedAvg and FedAvg. These outcomes highlight the effectiveness of the compression operators used in FedRIC and FedRICE, contributing to their high performance.

%%%%%%%%%

%%%%%%%%%%%%%%%%
\section{Conclusion}\label{sec:CO}
To address the optimization problem of minimizing communication costs and learning time in FL within the O-RAN environment, we formulated it as a Mixed-Integer Nonlinear Program~(MINLP), taking into account the heterogeneity in User Equipment~(UE) computing, traffic data, and channel conditions. This joint optimization of computing and radio resources considers stringent deadline services in the O-RAN RIC. We decomposed the problem into a UE Assignment~(UA) problem and a Computation and Bandwidth Allocation~(CBA) problem, providing a distributed solution. We presented FedRIC and FedRICE, novel Federated Learning~(FL) frameworks aimed at optimizing communication efficiency and latency within O-RAN components. These frameworks integrate randomized compression, latency-driven resource allocation, and multilevel hierarchical aggregation. Our experimental evaluations demonstrate that FedRIC and FedRICE outperform existing methods, such as FedCOM, FedAvg, and QFedAvg, in terms of convergence rate, communication cost, and latency.
\balance
%\newpage
\bibliographystyle{ieeetr} 
\bibliography{Refs/Ref,Refs/FL,Refs/infocom}

\end{document}